\documentclass[10pt]{article}
\usepackage{epsfig}
\usepackage{graphicx}
\usepackage{color}
\usepackage{amsmath,amsthm,amssymb}
\usepackage{latexsym}
\newtheorem{theor}{Theorem}[section]

\newtheorem{lemma}[theor]{Lemma}

\usepackage{hhline}

  \baselineskip 8 mm

\usepackage[pass, textheight=538pt,textwidth=318pt]{geometry}

\begin{document}

\title{Rainbow path and color degree in edge colored graphs}

\date{}

\author {Anita Das$^*$, P. Suresh$^\dag$, S. V. Subrahmanya \footnote{{E-Comm Research Lab, Education \& Research, Infosys Limited,
Bangalore, India. \newline \hspace{.5cm} $^\dag$E-Comm Research Lab,
Education \& Research, Infosys Limited, Mysore, India.
{\hspace{10cm}}{\bf Email:}(anita\_das01, suresh\_p01,
subrahmanyasv)@infosys.com } } }

\maketitle

\begin{abstract}
Let $G$ be an edge colored graph. A {\it}{rainbow path} in $G$ is a path in
which all the edges are colored with distinct colors.   Let $d^c(v)$ be the color degree of a vertex $v$ in $G$, 
i.e. the number of distinct colors present on the edges incident on the vertex $v$. Let $t$
be the maximum length of a rainbow path in $G$.  
Chen and Li showed that  
 if $d^c \geq k$, for every vertex $v$ of $G$, then $t \geq \left \lceil \frac{3 k}{5}\right \rceil + 1$
(Long heterochromatic paths in edge-colored graphs, The Electronic Journal of Combinatorics 12 (2005), \#  R33, Pages:1-33.)
Unfortunately, proof by Chen and Li is very long and comes to about 23 pages in the journal version.
Chen and Li states in their paper that it was conjectured by Akira Saito, that $t \ge \left \lceil \frac {2k} {3} \right \rceil$.
They also states in their paper that they believe $t \ge k - c$ for some constant $c$. 
 
In this note, we give a  short proof  to show that $t \ge \left \lceil \frac{3 k}{5}\right \rceil$, using
an entirely different method. Our proof is only
about 2 pages long. The draw-back is that our bound is less by 1, than the bound given by Chen and Li. 
We hope that the new approach adopted in this paper would eventually lead to the settlement of the conjectures by
Saito and/or Chen and Li. 
\vspace{.4cm}

\noindent{\bf keywords:} Edge colored graphs, rainbow path, color degree.
\end{abstract}

\section{Introduction}

Given a graph $G = (V, E)$, a map $c : E \rightarrow N$ ($N$ is the
set of non-negative integers) is called an {\it edge coloring} of
$G$. A graph G with  such a coloring $c$  is called
an {\it edge colored} graph. We denote the color of an edge $e \in
E(G)$ by $color(e)$. For a vertex $v$ of $G$, the {\it}{color
neighborhood $CN(v)$} of $v$ is defined as the set $\{color(e) | e$
is incident on $v\}$ and the {\it}{color degree} of $v$, denoted by
$d^c(v)$ is defined to be  $d^c(v) = |CN(v)|$.

A path in an edge colored graph with no two edges sharing the same
color is called a {\it rainbow path}. Similarly, a cycle in an edge
colored graph is called a {\it rainbow cycle} if no two edges of the
cycle share the same color. A survey on rainbow paths, cycles and
other rainbow sub-graphs can be found in \cite{kano}. Several
theorems and conjectures on rainbow cycles can be found in a paper
by Akbari, Etesami, Mahini and Mahmoody in \cite{AEMM}. 

Let $t$ denote the length of the maximum length rainbow path in $G$. In \cite{CL},
Chen and Li  studied the maximum length rainbow path problem in edge-colored graphs
and proved that if $G$ is an edge colored graph with $d^c(v) \geq k$, for every vertex $v$ of $G$, then $G$ has
a rainbow path of length at least $ \left \lceil \frac{3 k}{5}\right \rceil + 1$.
Chen and Li states in their paper that it was conjectured by Akira Saito, that $t \ge \left \lceil \frac {2k} {3} \right \rceil$.
They also states in their paper that they believe $t \ge k - c$ for some constant $c$, after showing
an example where the rainbow path cannot be more than $k-1$.

In this note, we give a  short proof  to show that $t \ge \left \lceil \frac{3 k}{5}\right \rceil$, using
an entirely different method. Our proof is only
about 2 pages long. The draw-back is that our bound is less by 1, than the bound given by Chen and Li. 
We hope that the new approach adopted in this paper would eventually lead to the settlement of the conjectures by
Saito and/or Chen and Li.

\subsection{Preliminaries}

All graphs considered in this paper are finite, simple and
undirected. A graph is a tuple $(V, E)$, where $V$ is a finite set
of vertices and $E$ is the set of edges. For a graph $G$, we use
$V(G)$ and $E(G)$ to denote its vertex set and edge set,
respectively. The neighborhood $N(v)$ of a vertex $v$ is the set of
vertices adjacent to $v$ but not including $v$. The degree of a
vertex $v$ is $d_v = |N(v)|$. A path is a non-empty graph $P=(V,E)$
of the form $V = \{p_1, p_2, \ldots , p_k\}$ and $E = \{ (p_1,
p_2), (p_2, p_3), \ldots , (p_{k-1}, p_k)  \}$, which we usually
denote by the sequence $\{p_1, p_2, \ldots, p_k\}$. The {\it length}
of a path is its number of edges. If $P = \{p_1, p_2, \ldots ,
p_k\}$ is a path, then the graph $C $ with $V(C) = V(P)$ and
$E(C)= P \cup \{(p_k, p_1)\}$ is a
cycle, and $|E(C)|$ is the length of $C$. We represent this cycle by
the cyclic sequence of its vertices, for example $C = \{p_1, p_2,
\ldots, p_k, p_1\}$.

\section{Proof of the main results}

Let $G$ be an edge colored graph with $d^c(v) \geq k$, for every vertex $v$ of $G$ and $t$ be the maximum length
rainbow path in $G$. Let  $C$ denote  the set of
colors used in the edge coloring of $G$. The
following lemma ensures a rainbow path of length
$\left\lceil\frac{k+1}{2}\right\rceil$ starting from any vertex
in an edge colored graph.

\begin{lemma}\label{lemma1}
Let $G$ be an edge colored graph and $d^c(v) \geq k$, for every $v
\in V(G)$. Then, given any vertex $x$ in $G$ there exists a rainbow
path of length at least $\left \lceil\frac{k+1}{2} \right \rceil$
starting from $x$.
\end{lemma}

\begin{proof}
Let $P$ be a rainbow path in $G$ of maximum length, say $t$. Thus
$P$ contains $t+1$ vertices, $t$ edges and hence $t$ distinct
colors. Let $P = \{x = u_0, u_1, \ldots,
u_t=y\}$, $U = \{color(u_i,u_{i+1}), 0 \leq i \leq t-1)$ and $U^c =
C \setminus U$.
Select a subset $E(x)$ of edges incident on $x$ as follows: Let
$(x,u_1) \in E(x)$. Now, select $k-1$ more edges incident 
on $x$ and add to $E(x)$ so that all the $k$ edges in $E(x)$ have
different colors. Clearly it is possible to do this, since $d^c(x) \ge k$.  

 Let  $N_P(x) = \{ u_i \ : \ 1 \leq i \leq t \
\mbox{ and } (x,u_i) \in E(x)\}$, $N_{P^c}(x) = \{ a \in V(G)
\setminus V(P) \ : \ (x, a) \in E(x)\}$.

\vspace{.2cm}

\noindent {\bf Claim.} If $z \in N_{P^c}(x)$, then $color(x,z) \in
U- \{ color(x,u_1) \} $.

Suppose that $color(x,z) \in U^c$. As $z \notin V(P)$, $P' = \{z,
x=u_0, u_1, \ldots, u_t=y\}$ is a path. Moreover, $P'$ is a rainbow
path as $color(x, z) \in U^c$. Now $P'$ is a rainbow path in $G$ of
length $t+1$. This is a contradiction to the fact that $t$ is the length of the
maximum length rainbow path in $G$.
Also it is obvious that $color(x,z) \ne color(x,u_1)$. Hence the claim is true.

From the above claim  we infer that
$N_{P^C} (x)| \leq |U| - 1 = t - 1$. So, $|N_P(x)| \geq
k - (t - 1)$. But, number of vertices in $P$
excluding $x$ is $t$. So, $k - (t - 1) \leq t$. Hence, $t \geq
\frac{k+1}{2}$. Since $t$ is an integer, we have $t \geq \left
\lceil \frac{k+1}{2} \right \rceil$.
\end{proof}

The following lemma ensures that if the maximum length of a rainbow
path is small enough, then we can convert the maximum rainbow path
into a rainbow cycle by some simple modifications.

\begin{lemma}\label{lemma2}
Let $G$ be an edge colored graph and $d^c(v) \geq k$, for every $v
\in V(G)$. Let $t$ be the  length of the maximum length rainbow path in $G$. If $t
< \left\lceil \frac{3}{5}k \right \rceil $, then $G$ contains a
rainbow cycle of length $(t+1)$.
\end{lemma}

\begin{proof}
Assume for contradiction that there is no rainbow cycle of length
$t+1$ in $G$. Let $P = \{u_0(=x), u_1, u_2, \ldots, u_{t}(=y)\}$  be
a rainbow path of length $t$ in $G$. Let $U = \{color(u_i,
u_{i+1})$, $0 \leq i \leq t-1\}$ and $U^c = C \setminus U$, where
$C$ is the set of colors used to color the edges of $G$. Clearly
$|U| = t$. Let $T_x = \{ u_i  \ : 0 \leq i \leq t, \ (x, u_i) \in
E(G)$ and $color(x, u_i) \in U^c\}$ and let $T_y = \{ u_i \ : \ 0
\leq i \leq t , \ (y, u_i) \in E(G)$ and $color(y, u_i) \in U^c\}$.

First note that, $|\{ (x,z) \in E(G): color(x,z) \in U^c \}| \ge k -
t$. Moreover, if $(x,z) \in E(G)$ with $color(x,z) \in U^c$, then $z
\in V(P)$, i.e., $z = u_i$ for some $1 \le i \le t$, since otherwise
we would have a rainbow path of length $t+1$ in $G$. It follows that
$|T_x| \geq k - t$. By a similar argument, we get $|T_y| \ge k - t$.
Note that $u_0, u_1 \notin T_x$ since $u_0=x$ and color($x,u_1) \in
U$. Also, $u_t \notin T_x$, since if $(x,u_t)$ is an edge and is
colored using a color from $U^c$, then we already have a $t+1$
length rainbow cycle, contrary to the assumption. So, we can write
$T_x = \{ u_i \ : 2 \le i \le t-1  \mbox { and } color(x,u_i) \in
U^c \}$. By similar reasoning, we can write, $T_y = \{ u_i \ : 1 \le
i \le t-2  \mbox { and } color(y,u_i) \in U^c \}$. Define $M_x =
\{u_j \ : \ u_{j+1} \in T_x\}$.

\vspace{.2cm}

\noindent  {\bf Observation 1.}
 $|M_x| = |T_x|  \geq k - t$.

\vspace{.2cm}

\noindent {\bf Claim 1.} $M_x \cap T_y \neq \emptyset$.

If possible suppose $M_x \cap T_y = \emptyset$. Now, $|M_x| + |T_y|
\leq t - 1$, as both $M_x \subset V(P)$ and $T_y \subset V(P)$  and
number of vertices on $P$ excluding $x$ and $y$ is $t-1$. (Note that,
$x,y \notin M_x$ and $x,y \notin T_y$.) As $|M_x| \geq k - t$
and $|T_y| \geq k - t$ and $M_x \cap T_y = \emptyset$ by
assumption, we have $k - t + k - t \leq t-1$. That is, $2 k \leq 3t - 1$. So, $t \geq \frac{2 k+1}{3}$. This is a
contradiction to the fact that $t < \left\lceil \frac{3}{5}k
\right\rceil $. Hence Claim 1 is true.

\vspace{.2cm}

\noindent {\bf Claim 2.} If $u_i \in M_x \cap T_y$, then $color(y,
u_i) = color(x, u_{i+1})$.

Suppose Claim 2 is false. That is, $ \exists u_i \in M_x \cap
T_y$ such that  $color(y, u_i) \neq color(x, u_{i+1})$. Now consider the
cycle: $CL = \{x, u_1, \ldots, u_i, y, u_t, u_{t-1},$ $\ldots,
u_{i+1}, x\}$. Clearly $CL$ is a rainbow cycle, as $color(y, u_i)
\neq color(x, u_{i+1})$,  $color(y, u_i) \in U^c$ and $color(x,
u_{i+1}) \in U^c$. Note that the length of $CL$ is $t+1$, as we
removed exactly one edge, namely $(u_i, u_{i+1})$ from $P$ and added
two new edges, namely $(y, u_i)$ and  $(x, u_{i+1})$ to $CL$. So, the
length of $CL$ is $t-1+2=t+1$, contradiction to the assumption.
Hence, we can infer that if $u_i \in M_x \cap T_y$, then $color(y,
u_i) = color(x, u_{i+1})$.

Let $S_y = \{  v \in V(P) - (M_x \cup \{ y, u_{t-1} \} )     \ : \
color(y, v) \in U  \}  \} $.

\vspace{.2cm}

\noindent {\bf Observation 2.}  $|M_x| + |T_y| + |S_y| - |M_x
\cap T_y| \leq t-1$.

Proof: This is because  $S_y$ is disjoint from $M_x \cup T_y$ and
$S_y \cup M_x \cup T_y \subseteq V(P) - \{ y, u_{t-1} \}$. (Note
that  $y (= u_t)$ and $u_{t-1}$ do not appear in $M_x$, $T_y$ or
$S_y$.)

\vspace{.1cm}

We partition the set $M_x \cap T_y$ as follows. Let $u_i \in M_x
\cap T_y$. If $color(u_i, u_{i+1})$ appears in one of the edges
incident on $y$, then $u_i \in A$ otherwise $u_i \in B$.

\vspace{.1cm}

\noindent {\bf Observation 3.}  $|T_y| \ge k - t + |B| $.

Proof: To see this first note that there are at least $k$ edges
of different colors
incident on $y$ (as by assumption, color degree of $y$ is at least
$k$)and at most $t- |B| $ of them can get the colors from $U$,
since $|B|$ colors in $U$ do not appear on the edges incident on
$y$, by the definition of $B$.
So, at least $k - t
+ |B| $ of the edges incident on $y$ have colors from $U^c$, and
clearly any $w$, such that $(y,w)$ is an edge, colored by a color in
$U^c$ has to be on $P$, since otherwise we have a longer rainbow
path. It follows that $|T_y| \ge k - t + |B|$.

\vspace{.1cm}

\noindent{\bf Claim 3.} If $u_i \in A$, then the edge incident on
$y$ with color $color(u_i, u_{i+1})$ has its other end point on the
rainbow path $P$. That is, if $w$ is such that $(y,w)$ is an edge
and $color(u_i,u_{i+1}) = color(y, w)$, then $w \in V(P)$.

Suppose Claim 3 is false. Let $(y, w) \in E(G)$ with
$color(y, w) = color(u_i, u_{i+1})$ and $w \notin V(P)$. Now,
consider the path: $P' = \{ w, y, u_{t-1}, u_{t-2}, \ldots, u_{i+1},
x (=u_0), u_1, \ldots, u_i\}$. Clearly $P'$ is a rainbow path as
$color(u_i,u_{i+1}) = color(y, w)$,  the edge $(u_i,u_{i+1}) \notin
E(P')$ and color$(u_{i+1},x) \in U^c$, since $u_i \in M_x$. Note
that, the length of $P'$ is $t+1$. This is a contradiction to the
fact that $t$ is the maximum length rainbow path in $G$. Hence Claim
3 is true.

\vspace{.1cm}

Now, partition $A$ as follows: if $u_i \in A$, then by the above
claim the edge incident on $y$ with the color $color(u_i, u_{i+1})$
has its other end point say $w$,  on $P$. If $w \in M_x$, then let $u_i
\in A_1$, else $u_i \in A_2$.

\vspace{.1cm}

\noindent {\bf Observation 4.}  $|M_x \cap T_y| = |A| + |B| = |A_1|
+ |A_2| + |B|$.

\vspace{.1cm}

\noindent {\bf Observation 5.}  $|S_y| \ge |A_2|$. To see this,
recall that $S_y = \{v \in V(P) -  (M_x \cup \{ y, u_{t-1} \})   \ : \ color(y, v) \in U 
 \} $. By definition of $A_2$, for each $u_i \in
A_2$ there exists a unique  vertex $w = w(u_i) \in V(P) - M_x$ such
that $(y, w)$ is an edge and $color(y,w) = color(u_i,u_{i+1}) \in
U$. Since $u_i \in A_2 \subset M_x$, we have $i <  t-1$ and thus
color$(u_i,u_{i+1}) \ne color(y,u_{t-1})$. Therefore $w(u_i)$ cannot be $y$ or $u_{t-1}$, for any $u_i \in A_2$.
 It follows that $\{w(u_i)
: u_i \in A_2 \} \subseteq  S_y$, and therefore we have $|S_y| \ge
|A_2|$.

\vspace{.2cm}

\noindent {\bf Claim 4.}   $|A_1| \leq \frac{|M_x|}{2}$.

Recall that, for each $u_i \in A_1$, there is a unique vertex
$w=w(u_i)$ such that $(y,w)$ is an edge with  $color(u_i, u_{i+1})=
color(y,w)$. Moreover, $w \in M_x$, by the definition of $A_1$ and
$A_1 \cup \{ w(u_i) : u_i \in A_1 \} \subseteq M_x$. Note that
$w(u_i)$ is uniquely defined for $u_i$ since it is the end point of
the edge incident on $y$ colored with the color of the edge
$(u_i,u_{i+1})$. Moreover,  $A_1 \cap \{ w(u_i) : u_i \in A_1 \} =
\emptyset$, since $A_1$ contains vertices which are end points of
edges from $y$, colored by the colors in $U^c$ whereas each $w(u_i)$
is the end point of some edge from $y$ which is colored by a color
in $U$. It follows that $2|A_1| \leq |M_x|$. That is, $|A_1 \leq
\frac{|M_x|}{2}$, as required.

Now, substituting $k - t + |B| $ for $|T_y|$ (by Observation 3),
$|A_2|$ for $|S_y|$ (by Observation 5), and $|A_1| + |A_2| + |B|=
|M_x \cap T_y|$ (by Observation 4) in  the  inequality of
Observation 2,  and simplifying we get $|M_x| + k - t  - |A_1|
\leq t - 1 $. Now using $|A_1| \le |M_x|/2$ (Claim 4) and
and simplifying we get $\frac{|M_x|}{2} + k - t  \leq t - 1$. Recall
that $|M_x| \geq k -t$ (Observation 1). Substituting and simplifying
we get, $t \geq \frac{3k + 2}{5}$. It follows that $t  \ge \left\lceil \frac{3}{5} k
\right \rceil $, contradicting the initial assumption. Hence the
Lemma is true.
\end{proof}

\begin{theor}
Let $G$ be an edge colored graph and $d^c(v) \geq k$, for every $v
\in V(G)$. If $t$ is the maximum length of a rainbow path in $G$,
then $t \geq \left \lceil\frac{3k}{5} \right\rceil$.
\end{theor}

\begin{proof}
If possible suppose $t < \left \lceil \frac{3
k}{5}\right\rceil $. By Lemma \ref{lemma2}, $G$ contains a
rainbow cycle of length $t+1$. Let $CL$ be this cycle. Note that,
$CL$ contains $(t+1)$ vertices and $(t+1)$ edges. Now, $t+1 \leq
\lceil \frac{3 k}{5} \rceil  $. Let $CL = \{u_0, u_1, \ldots,
u_t,  u_0\}$ and $V(CL^c) = V(G) \setminus V(CL)$. Let $U =
\{color(e) \ : \ e \in E(CL)\}$ and $U^c = C \setminus U$, where $C$
is the set of colors used to color the edges of $G$. Let
$F_i = \{ z \in V(CL^c) : (u_i,z) \in E(G)\}$.

\vspace{.1cm}

\noindent {\bf Claim 1.} $|F_i| \ge \left\lfloor\frac {2 k}
{5}\right\rfloor $. Moreover, for $z \in F_i$, $color(u_i, z) \in
U$.

First part follows from the fact that the  degree of $u_i \ge 
d^c(u_i) \geq k$ and there are only at most $\left \lceil\frac {3
k} {5} \right \rceil $ vertices in $CL$ and therefore $u_i$
has only at most $\left \lceil \frac {3k} {5} \right \rceil - 1$
neighbours in CL. If possible suppose
$color(u_i, z) \in U^c$. Now consider the path $P' = \{z, u_i,
u_{i+1}, u_{i+2}, \ldots,$ $u_{t},u_0,\ldots, u_{i-1} \}$. Clearly,
$P'$ is a rainbow path as $ color(u_i, z) \in U^c$ and $\{u_i,
u_{i+1}, u_{i+2}, \ldots, u_{t},$ $u_0, \ldots, u_{i-1} \}$ is
already a rainbow path being a part of the rainbow cycle $CL$. Note
that, the length of $P'$ is $t+1$. This is a contradiction to the
assumption that $t$ is the length of the 
 maximum length rainbow path in $G$. Hence
Claim 1 is true.

\vspace{.1cm}

Let $G' = (V', E')$, where $V'(G') = V(G)$ and $E'(G') = E(G)
\setminus \{e \in E(G) : \  color(e) \in U\}$. Clearly, in $G'$
there is no edge between $V(CL)$ to $V(CL^c)$, since by Claim 1,
every such edge is colored by a color in $U$. 
 Let $G'' = G'[V(CL^c)]$
be  the induced
subgraph on $V(CL^c)$ in $G'$ and let $k'$ be the minimum color
degree of $G''$.

\vspace{.1cm}

\noindent {\bf Observation 1.} $k' \ge \left\lfloor \frac {2k}
{5}\right\rfloor  - 1$.

Proof: Clearly  $k' \geq k - |U| = k - (t+1) \ge k - \left\lceil
\frac{3 k}{5}\right\rceil  \geq \left\lfloor
\frac{2k}{5}\right\rfloor  - 1$.

\vspace{.2cm}


\noindent Consider the following subset $U_0$ of $U$, defined by
$U_0 = U_1 \cup U_2$, where $U_1 = \{ color(u_i,u_{i+1}) : 0 \le i
\le \lceil \frac {k} {5} \rceil  -1\}$ and $U_2 =  \{ color (u_i,
u_{i+1}) : (t+1) - \left \lceil \frac {k} {5} \right \rceil \le i
\leq t-1 \} \cup \{color(u_t, u_0)\} $.

\vspace{.1cm}

\noindent {\bf Claim  2.}  $ \{ color(u_0,z) : z \in F_0\}
 \cap U_0 = \emptyset$.

Suppose not. Let $z \in F_0$ be such that $color(u_0,z)  \in U_0$.
Without loss of generality assume that $color(u_0,z) \in U_1 $. Then
consider the path $P^*$ = $(u_{\left \lfloor \frac {k} {5} \right
\rfloor}, \ldots, u_{t}, u_0,z)$, which is clearly a rainbow path,
since the edge of $CL$ with its color equal to $color(u_0,z)$ is not
there in this path. Also the length of $P^*$ is $t+2 - \left \lfloor
\frac {k} {5} \right \rfloor$. By Observation 1, $G''$ has minimum color
degree at least $\left\lfloor \frac {2 k} {5}\right\rfloor  - 1$, and
therefore by Lemma \ref {lemma1}, $G''$ has a rainbow path of length
at least $\left\lceil\frac {k} {5}\right\rceil - 1$, let us call this
path $P''$. Clearly concatenating the path $P''$ with $P^*$ we get a
rainbow path since colors used in $P^*$ belong to $U$ whereas the
colors used in $P''$ belong to $U^c$. Moreover, the length of this
rainbow path is at least $t+1$, a contradiction, to the assumption
that $t$ is the length of the maximum rainbow path in $G$.

Now we complete the proof as follows: In view of Claim 2, and Claim
1, we know that  the number of colors in $F_0$ is $ \leq |U - U_0|$. 
But $|U - U_0| \le
\left\lceil\frac{3 k} {5}\right\rceil  - ( 2 \left\lceil\frac { k }
{5}\right\rceil   )\leq \frac {3k} {5} + 1 - \frac {2k} {5} 
 \leq \left\lceil\frac{k}{5}\right\rceil + 1  <
\left\lceil\frac {2 k} {5}\right\rceil$ (since we can assume $k > 
5$: for smaller values of $k$, the Theorem is trivially true). This
is a contradiction to the first part of Claim 1.
\end{proof}

\end{document}